\documentclass[twocolumn,pra,letterpaper]{revtex4}

\usepackage{hyperref}
\usepackage{latexsym,amsmath,amssymb,amsfonts,graphicx,color,amsthm}
\usepackage{enumerate}
\usepackage{epstopdf}

\newcommand{\ket}[1]{|#1\rangle}
\newcommand{\bra}[1]{\langle #1|}
\newcommand{\state}[1]{\ket{#1}\!\bra{#1}}

\newcommand{\braket}[2]{\langle #1 | #2\rangle}
\newcommand{\matel}[3]{\langle #1|#2|#3\rangle}

\newcommand{\Tr}{\text{Tr}}
\newcommand{\tr}{\text{Tr}}
\newcommand{\Id}{\mathbb{I}}
\newcommand{\cE}{\mathcal{E}}
\newcommand{\cM}{\mathcal{M}}
\newcommand{\cB}{\mathcal{B}}

\newcommand{\cS}{\mathcal{S}}

\newcommand{\cH}{\mathcal{H}}
\newcommand{\cQ}{\mathcal{Q}}
\newcommand{\cR}{\mathcal{R}}

\newcommand{\mx}{\text{max}}
\newcommand{\acc}{\text{acc}}

\newtheorem{lemma}{Lemma}
\newtheorem{theorem}{Theorem}
\newtheorem{corollary}{Corollary}
\newtheorem{definition}{Definition}

\begin{document}
\title{Quantifying incompatibility beyond entropic uncertainty}
\author{Srijita Kundu}
\affiliation{Chennai Mathematical Institute, Siruseri, Chennai   - 603103, India.}
\author{Kaustubh Wagh}
\affiliation{Department of Physics, Indian Institute of Technology Madras, Chennai - 600036, India. }
\author{Prabha Mandayam}
\affiliation{Department of Physics, Indian Institute of Technology Madras, Chennai - 600036, India. }

\date{\today}

\begin{abstract}
We study two operational approaches to quantifying incompatibility that depart significantly from the well known entropic uncertainty relation (EUR) formalism. Both approaches result in incompatibility measures that yield non-zero values even when the pair of incompatible observables commute over a subspace, unlike EURs which give a zero lower bound in such cases. Here, we explicitly show how these measures go beyond EURs in quantifying incompatibility: For {\it any} set of quantum observables, we show that both incompatibility measures are bounded from below by the corresponding EURs for the Tsallis ($T_{2}$) entropy. We explicitly evaluate the incompatibility of a pair of qubit observables in both operational scenarios. We also obtain an efficiently computable lower bound for the mutually incompatibility of a general set of observables. 
\end{abstract}

\maketitle

Characterizing the mutual incompatibility of a set of quantum observables is an important question both from a quantum cryptographic as well as a foundational point of view. While the Heisenberg uncertainty principle~\cite{heisenberg:ur, robertson:uncertainty} provided the first quantitative statement on the incompatibility of a pair of canonically conjugate observables, later formulations in terms of entropic quantities characterize the incompatibility of any set of observables via entropic uncertainty relations (EURs) (see~\cite{ww:urSurvey} for a recent review). 

For a given set of observables, an EUR seeks to lower bound the sum of the entropies associated with the probability distributions arising from measurements of these observables, on distinct yet identically prepared copies of a quantum system. EURs play an important role in the security of quantum cryptographic tasks~\cite{qkd:ur, renes:ur, prl:noisy}, and are often thought to quantify the incompatibility of a set of quantum observables. However, EURs give a trivial lower bound whenever the observables in question share a single common eigenvector, although there are variants of the standard EUR formalism that circumvent this pitfall in certain specific cases~\cite{fine-grained, univ_UR, EUR_nonlocality13}. Uncertainty lower bounds cannot therefore be considered a {\it measure} of incompatibility in general.

This has motivated the emergence of alternate approaches to quantifying incompatibility. 
The measure $\cQ$ defined in~\cite{bandyoPM_PRA13} is based on the fact that the eigenstates associated with a set of incompatible observables are not perfectly distinguishable. On the other hand, the class of measures $\{\cQ_{\alpha}, \alpha = 1, F, \infty\}$ defined in~\cite{PM_MDS14} quantify the mutual incompatibility of a pair of observables by estimating the disturbance due to a measurement of one observable on the statistics of the outcomes of the other. Both these measures have the desirable property that they are strictly zero if and only if the observables in the set all commute. While bounds on these measures are known, exact expressions have been obtained only for a set of mutually unbiased observables~\cite{bandyoPM_PRA13, PM_MDS14}. Evaluating these measures exactly for any set of observables is in general a hard optimization problem to which an efficient solution is not yet known.

Here, we clarify the exact relation between these incompatibility measures and the standard EUR formalism. For a general set of observables we prove that the measure $\cQ$ is bounded from below by the lower bound on the corresponding Tsallis ($T_{2}$) entropic uncertainty relation, and, the measure $\cQ_{F}$ is bounded from below by the the $T_{2}$ entropic lower bound in a successive measurement scenario. We also evaluate the incompatibility of a pair of observables that commute on a subspace, thus providing an explicit example of a class of observables for which these measures go beyond EURs.  Finally, we obtain efficiently computable lower bounds for both the incompatibility measures for {\it any} set of observables.


The rest of the paper is organized as follows. We present a brief review of the incompatibility measures $\cQ$ and $\{\cQ_{\alpha}\}$ in Sec.~\ref{sec:Incompat}.  We demonstrate the relations between these measures and Tsallis EUR lower bounds in Sec.~\ref{sec:Q_EUR}. We evaluate the incompatibility of a pair of qubit observables in Sec.~\ref{sec:Q_qubit}, and consider the case of observables that commute on a subspace in Sec.~\ref{sec:commute}. Finally, in Sec.~\ref{sec:Q_lbound} we obtain a lower bound on the incompatibility of a general set of observables, which can be efficiently computed via convex optimization.

\section{Operational measures of incompatibility} \label{sec:Incompat}

We begin with a brief review of two operational approaches to quantifying incompatibility proposed recently~\cite{bandyoPM_PRA13, PM_MDS14}. Throughout this paper we work within the framework of standard quantum theory, and restrict our attention to observables associated with self-adjoint operators with discrete spectra, on a $d$-dimensional Hilbert space $\cH_{d}$. 

\subsection{Distinguishability-based measure}\label{sec:F_acc}

An important physical manifestation of the mutual incompatibility of a set of observables is the fact that the eigenstates of such a set are not all mutually orthogonal and therefore cannot be distinguished perfectly. This motivates the incompatibility measure $\cQ$ defined in~\cite{bandyoPM_PRA13}, which is based on quantifying the extent to which the eigenstates corresponding to a given set of observables are not distinguishable.  

Consider a set of $N$ non-degenerate observables $\{A^{(1)}, A^{(2)}, \ldots, A^{(N)}\}$ on a $d$-dimensional Hilbert space $\cH_{d}$, and let $\vert a^{(i)}_{j}\rangle\langle\ a^{(i)}_{j}\vert$ denote the $j^{\rm th}$ eigenstate of the $i^{\rm th}$ observable $A^{(i)}$. If the observables $A^{(1)}, A^{(2)}, \ldots, A^{(N)}$ do not all commute, they do not have a complete set of common eigenstates, and hence at least some of their eigenstates must be non-orthogonal. Then, the incompatibility measure $\cQ(A^{(1)}, A^{(2)}, \ldots, A^{(N)})$ defined in~\cite{bandyoPM_PRA13} quantifies the extent to which the states $\{\vert a^{(i)}_{j}\rangle\langle a_{j}^{(i)}\vert\}$, $i\in \left[N\right]$, $j\in \left[d\right]$ are not distinguishable. The more incompatible the observables $\{A^{(i)}\}$ are, the lesser is the fidelity with which their eigenstates can be distinguished. The incompatibility $\cQ$ of a set of observables is thus defined as the {\it complement} of the best possible fidelity obtained in a quantum state estimation process~\cite{navascues} for the uniform ensemble of their eigenstates.

Specifically, let $\cS \equiv \{|a^{(i)}_{j}\rangle\langle a^{(i)}_{j}|\}$, $i\in\left[N\right]$, $j\in\left[d\right]$ denote the ensemble comprising eigenstates of the observables $\{A^{(i)}\}$, wherein the states are all picked with equal probability $\frac{1}{Nd}$. The quantum state estimation problem for the ensemble $\cS$ seeks to maximize the {\it average fidelity} function
\begin{equation}
F_{\cS}(\cM, \cR) = \frac{1}{Nd}\sum_{k,i,j} \langle a_{j}^{(i)}| M_{k}|a^{(i)}_{j}\rangle \langle a_{j}^{(i)}| \sigma_{k}|a^{(i)}_{j}\rangle , \nonumber
\end{equation}
over all possible positive operator valued measures (POVM) $\cM$ (with elements $\{M_{k}\}$, $0\leq M_{k} \leq \Id$, $\sum_{k} M_{k} = \Id$), and all state reconstruction maps $\cR: k \rightarrow \sigma_{k}$ (in which the state $\sigma_{k}$ is the estimate corresponding to measurement outcome $k$). The maximum fidelity that can obtained for the ensemble $\cS$ in a state estimation process is given by,
\begin{eqnarray}
&& F^{\rm max}_{\cS} \nonumber \\
& = & \frac{1}{Nd}\sup_{\cM, \cR} \sum_{k,i,j}\langle a_{j}^{(i)}| M_{k}| a^{(i)}_{j}\rangle \langle a_{j}^{(i)}| \sigma_{k}| a^{(i)}_{j}\rangle. \label{eq:maxFid}
\end{eqnarray}
The mutual incompatibility $\cQ$ of the observables $\{A^{(1)}, \ldots, A^{(N)}\}$ is then defined as~\cite{bandyoPM_PRA13},
\begin{equation}
\cQ(A^{(1)}, \ldots, A^{(N)}) = 1 - F^{\rm max}_{\cS}. \label{eq:Q}
\end{equation}
It is easy to see that $0 \leq \cQ(A^{(1)}, \ldots, A^{(N)}) \leq 1$ for any set of observables $\{A^{(1)}, A^{(2)}, \ldots, A^{(N)}\}$, with the lower bound being attained if and only if the observables $\{A^{(1)}, \ldots, A^{(N)}\}$ all commute. 

The measure $\cQ$ is also of direct relevance in quantum cryptography, specifically, in the context of quantum key distribution (QKD) protocols of the {\it prepare and measure} type~\cite{qkd_rmp}. As noted in~\cite{bandyoPM_PRA13}, the measure $\cQ(A^{(1)}, \ldots, A^{(N)})$ is simply the complement of the accessible fidelity~\cite{Fuchs-Sasaki, fuchs04} of the ensemble $\cS$, that is, the best possible fidelity an  eavesdropper can obtain by employing an ``intercept-resend strategy'', in a QKD protocol where Alice transmits pure states $|a^{(i)}_{j}\rangle\langle a^{(i)}_{j}|$ drawn uniformly at random from the ensemble $\cS$. 

In fact there exists a more direct, quantitative relation between the incompatibility measure $\cQ(A^{(1)}, A^{(2)}, \ldots, A^{(N)})$ and the error rate caused by the presence of an eavesdropper in a QKD protocol whose signal ensemble comprises the eigenstates of the observables $\{A^{(1)}, A^{(2)}, \ldots, A^{(N)}\}$. We make this precise in Appendix~\ref{sec:Q_QKD}.


\subsection{Distance-based incompatibility measures} \label{sec:incompat_dist}

Another important consequence of the mutual incompatibility of a pair of observables $A,B$ is the change caused to the measurement statistics of one (say $B$) due to an intervening measurement of another (say $A$) on the {\it same} state. Quantifying this change in measurement statistics, also leads to a characterization of the incompatibility of the pair of observables $A$ and $B$, as shown in~\cite{PM_MDS14}.

For a pair of observables $A, B$ with spectral decompositions $A= \sum_{i=1}^{d}a_{i}P^{A}_{i}$ and $B=\sum_{j=1}^{d} b_{j}P^{B}_{j}$, consider the following two probability distributions: one resulting from a measurement of $A$ followed by a measurement of $B$ on a given state, and the other obtained from a measurement $B$ alone on the same state. Let ${\rm Pr}^{B}_{\rho} \sim \{p_{\rho}^{B}(j), \; j = 1,\ldots, d\}$ denote the probability distribution over the outcomes of a measurement of observable $B$ in state $\rho \in \cB(\cH_{d})$. Let ${\rm Pr}^{A\rightarrow B}_{\rho} \sim \{q^{A\rightarrow B}_{\rho}(j), \; j=1,\ldots, d\}$ denote the probability distribution over the outcomes of a $B$ measurement following a measurement of $A$ on the same state $\rho$. These probabilities are simply given by
\begin{eqnarray}
{\rm Pr}^{B}_{\rho}: p^{B}_{\rho}(j) &=& \tr\left[P^{B}_{j}\rho\right], \; j = 1,\ldots, d.   \label{eq:prob_defn}\\
 {\rm Pr}^{A\rightarrow B}_{\rho} : q^{A\rightarrow B}_{\rho}(j) &=& \tr\left[P^{B}_{j}\left(\sum_{i=1}^{d}P^{A}_{i}\rho P^{A}_{i}\right)\right]. \nonumber
\end{eqnarray}
If $A$ and $B$ commute, their corresponding eigen-projectors commute, and the two probability distributions defined in Eq.~\eqref{eq:prob_defn} above are identical for all states $\rho$. For a general pair of observables, the distance between the probability distributions ${\rm Pr}^{A\rightarrow B}_{\rho}$ and  ${\rm Pr}^{B}_{\rho}$ can thus be regarded as a measure of how much a measurement of $A$ disturbs the statistics of the outcomes of a subsequent measurement of $B$ on the same state $\rho$. It was shown that maximizing the distance between these two distributions over all states yields a valid measure of the incompatibility of observables $A$ and $B$, which is zero if and only if $A$ and $B$ commute and is strictly greater than zero otherwise~\cite{PM_MDS14}.


Corresponding to the standard classical distance measures between probability distributions~\cite{Cha, Mathai_Rathie}, the following three measures of incompatibility of observable $A$ with $B$ were defined in~\cite{PM_MDS14}.
\begin{itemize}
 \item[(i)] $L_{1}$-distance based incompatibility measure:
\begin{equation}
\cQ_{1}(A\rightarrow B) = \sup_{\rho} D_{1}\left({\rm Pr}^{A\rightarrow B}_{\rho},{\rm Pr}^{B}_{\rho} \right), \nonumber
\end{equation}
where, $D_{1}(P,Q) \equiv \frac{1}{2}\sum_{i}\vert p_{i} - q_{i}\vert$ is the $L_{1}$-distance.
\item[(ii)] Fidelity-based incompatibility measure:
\begin{equation}
 \cQ_{F}(A\rightarrow B) = \sup_{\rho}\left[1  - F^{2}({\rm Pr}^{A\rightarrow B}_{\rho},{\rm Pr}^{B}_{\rho})\right], \nonumber
\end{equation}
where $F(P,Q) \equiv \sum_{i}\sqrt{p_{i}}\sqrt{q_{i}}$ is the fidelity, also known as the Bhattacharyya distance. 
\item[(iii)] $L_{\infty}$-distance based incompatibility measure:
\begin{equation}
 \cQ_{\infty}(A\rightarrow B) = \sup_{\rho} D_{\infty}({\rm Pr}^{A\rightarrow B}_{\rho},{\rm Pr}^{B}_{\rho}), \nonumber
\end{equation}
where $D_{\infty}(P,Q) \equiv \max_{i}\vert p_{i} -  q_{i}\vert$ is the Chebyshev or $L_{\infty}$-distance. 
\end{itemize}

Note that all three incompatibility measures defined above satisfy,
\begin{equation}
 0 \leq \cQ_{\alpha}(A\rightarrow B) \leq 1, \;\;\; \alpha, \in \{1, F, \infty\}
\end{equation}
where the lower bound is attained if and only if the observables $A$ and $B$ commute~\cite{PM_MDS14}. Furthermore, $\cQ_{\alpha}(A \rightarrow B) \neq \cQ_{\alpha}(B \rightarrow A)$ in general. The incompatibility $\cQ_{\alpha}(A,B)$ of the pair of observables $A,B$ is therefore defined as the average of the incompatibilities $\cQ_{\alpha}(A\rightarrow B)$ and $\cQ_{\alpha}(B\rightarrow A)$, thus ensuring that $\cQ_{\alpha}(A,B)$ is large when {\it both} $\cQ_{\alpha}(A\rightarrow B)$ and $\cQ_{\alpha}(B\rightarrow A)$ are large and vice-versa. Finally, the incompatibility of a set of $N$ observables $\{A_{1}, A_{2}, \ldots, A_{N}\}$ is defined in terms of the pairwise incompatibilities $\{\cQ_{\alpha}(A_{i}\rightarrow A_{j})\}$, as, 
\begin{equation}
\cQ_{\alpha}(A_{1}, A_{2}, \dots, A_{N}) \equiv \frac{1}{N^{2}}\sum_{i,j}\cQ_{\alpha}(A_{i}\rightarrow A_{j}), \label{eq:sym_measureN}
\end{equation}
where $\cQ_{\alpha}(A_{i}\rightarrow A_{i}) = 0$. In the rest of the paper, we restrict our attention to the fidelity-based measure $\cQ_{F}$.

\section{Incompatibility and Entropic Uncertainty}\label{sec:Q_EUR}

While the operational approaches presented above depart significantly from the standard entropic uncertainty based approach to quantifying incompatibility, it is useful to understand how these two approaches are related. 

We begin with a brief review of the EUR formalism. We restrict our attention to observables associated with a non-degenerate spectra. For a set of $N$ observables $\{A^{(1)}, A^{(2)}, \ldots, A^{(N)}\}$, let $|a^{(i)}_{j}\rangle\langle a^{(i)}_{j}|$ denote the $j^{\rm th}$ eigenstate of the observable $A^{(i)}$. Then, an entropic uncertainty relation (EUR) seeks to lower bound the average of the entropies associated with a measurement of each (on {\it distinct} yet {\it identically prepared states}), as follows:
\begin{equation}
\inf_{\ket{\phi}} \frac{1}{N} \sum_{i=1}^N S(A^{(i)}; \ket{\phi}) \ge c(A^{(1)},\ldots, A^{(N)}), \label{eq:eur}
\end{equation}
where $S(A^{(i)}; \ket{\phi})$ denotes some entropy function of the probability distribution $p_{|\phi\rangle}^{A^{(i)}}(j) = |\braket{a^{(i)}_{j}}{\phi}|^{2}$ arising from a measurement of $A^{(i)}$ on state $|\phi\rangle$. The commonly used entropy functions belong to the R\'enyi class of entropies~\cite{renyi:entropy}  or the Tsallis class of entropies~\cite{tsallis:entropy}. For any $\alpha \geq 0$, the R\'enyi entropy $H_{\alpha}(.)$ and the Tsallis entropy $T_{\alpha}(.)$ of order $\alpha$ are defined as follows:
\begin{eqnarray}
H_{\alpha}(\{ p(i) \}) &:=&  \frac{1}{1-\alpha}\log\left(\sum_{i}(p(i))^{\alpha}\right), \nonumber \\
T_{\alpha}(\{p(i)\}) &:=& \frac{1}{1-\alpha}\left(\sum_{i}(p(i))^{\alpha} -1\right) , \nonumber
\end{eqnarray}
Throughout this paper we use $\log$ to denote the logarithm taken to the base $2$. In the limiting case of $\alpha\rightarrow 1$, $H_{1}(.) = T_{1}(.)$ is the well known Shannon entropy. We note that the R\'enyi entropies $H_{\alpha}$ are concave for $ 0 < \alpha \leq \alpha* = 1 + \frac{2}{\log (d-1)} $, and the Tsallis entropies $T_{\alpha}$ are concave for $\alpha >0$.

For any set of observables $\{A^{(1)}, A^{(2)}, \ldots, A^{(N)}\}$, the entropic lower bound in Eq.~\eqref{eq:eur} corresponding to the R\'enyi class of entropies satisfies
\[ 0 \leq c_{\alpha}(A^{(1)},\ldots, A^{(N)}) \leq \left(\frac{L-1}{L}\right)\log d  \; (\forall \alpha \geq 0), \]
whereas the Tsallis entropic lower bounds satisfy,
\[ 0 \leq t_{\alpha}(A^{(1)},\ldots, A^{(N)}) \leq \left(\frac{L-1}{L}\right) \frac{d ^{1-\alpha} -1}{1-\alpha} \; (\forall \alpha \geq 0). \]
In both cases, the trivial (zero) value is attained for any common eigenstate of the observables $\{A^{(1)}, \ldots, A^{(N)}\}$. We refer to~\cite{ww:urSurvey} for a recent survey of the known entropic uncertainty relations for different sets of observables. 

EURs have also been formulated for the case where a pair of observables $A$ and $B$ are measured {\it sequentially} on a system in state $|\psi\rangle$.  Here, the uncertainty in the outcome of the first measurement, say observable $A$, is given as before by $S(A;\ket{\psi})$, but the uncertainty in the outcome of a subsequent measurement of $B$ is to be calculated with respect to the post-measurement state 
\[ \mathcal{E}^{A}(|\psi\rangle) = \sum_{i} P^{A}_{i}|\psi\rangle\langle\psi| P^{A}_{i}, \]
and is therefore denoted as $S (B;\mathcal{E}^{A}(|\psi\rangle))$. Entropic uncertainty relations in the successive measurement scenario are therefore of the form~\cite{MDS01, MDS03_successiveEUR},
\begin{eqnarray}
\frac{1}{2}\inf_{\ket{\psi}} \left[ S(A; \ket{\psi}) + S (B; \cE^{A}(\ket{\psi}))\right] &=& c (A\rightarrow B), 
\nonumber  \\
\frac{1}{2}\inf_{\ket{\psi}} \left[ S(B; \ket{\psi}) + S (A; \cE^{B}(\ket{\psi}))\right] &=& c (B\rightarrow A). 
\label{eq:succEUR}
\end{eqnarray}
We note that the lower bound may change depending on which of the observables $A$ or $B$ is measured first. Such EURs have been studied for the Shannon entropy ($H_{1}$) for a general pair of observables~\cite{succEUR_zhang15} and the Tsallis class of entropies for a pair of qubit observables~\cite{Rastegin15}. In fact, when the entropy function is concave, there exists an explicit closed form expression for the lower bounds in Eq.~\eqref{eq:succEUR}, as we note below.
\begin{lemma}\label{lem:succEUR}
Consider a pair of observables $A,B$ with non-degenerate spectra and eigenvectors $\{|a_{i}\rangle\}$ and $\{|b_{j}\rangle\}$ respectively. In the case of a successive measurement of $A$ followed by a measurement of $B$ on the same system, the following optimal entropic uncertainty relations hold.
\begin{eqnarray}
&& \frac{1}{2}\inf_{\rho} \left[ H_{\alpha}(A;\rho) + H_{\alpha}(B;\cE^{A}(\rho))\right] \nonumber \\
&=& \inf_{i} \frac{1}{1-\alpha}\log\left(\vert\sum_{j=1}\vert\langle a_{i}|b_{j}\rangle\vert^{2}\vert^{\alpha}\right), \; (0\leq \alpha \leq \alpha^{*} ), \nonumber \\
&& \frac{1}{2}\inf_{\rho} \left[T_{\alpha}(A;\rho) + T_{\alpha}(B; \cE^{A}(\rho))\right] \nonumber \\
&=& \inf_{i} \frac{1 - \sum_{j}\vert\langle a_{i}|b_{j}\rangle\vert^{2\alpha} }{\alpha -1}.  \; (\alpha > 0)
\end{eqnarray}
\end{lemma}
\begin{proof}
The concavity of the R\'enyi and Tsallis entropies for the ranges $0< \alpha \leq \alpha^{*}$ and $\alpha>0$ respectively implies that lower bound will be attained for pure states. Furthermore, as noted in~\cite{MDS03_successiveEUR}, $T_{\alpha}(A;\rho) = T_{\alpha}(A; \cE^{A}(\rho))$ and $H_{\alpha}(A;\rho) = H_{\alpha}(A;\cE^{A}(\rho))$, for all states $\rho$. Hence, the optimization on the LHS is to be done over pure states of the form $\cE^{A}(\rho)$, for some $\rho$. In other words, we only need to minimize over the eigenstates $\{|a_{i}\rangle\}$ of $A$. Therefore,
\begin{eqnarray}
&& \inf_{\rho} \frac{1}{2}\left[T_{\alpha}(A;\rho) + T_{\alpha}(B; \cE^{A}(\rho)) \right] \nonumber \\
&=& \inf_{i} T_{\alpha}(B; |a_{i}\rangle) = \inf_{i}\frac{1 - \sum_{j=1}^{d}\vert\langle a_{i}|b_{j}\rangle\vert^{2\alpha}}{\alpha -1} . \nonumber
\end{eqnarray}
Similarly, 
\begin{eqnarray}
&& \inf_{\rho} \frac{1}{2}\left[H_{\alpha}(A;\rho) + H_{\alpha}(B; \cE^{A}(\rho)) \right] \nonumber \\
&=&  \inf_{i} H_{\alpha}(B; |a_{i}\rangle) = \inf_{i}\frac{\log\vert\sum_{j=1}^{d}\vert\langle a_{i}|b_{j}\rangle\vert^{2}\vert^{\alpha}}{1-\alpha}. \nonumber
\end{eqnarray}
\end{proof} 

In the rest of the paper we restrict our attention to the Tsallis and R\'enyi entropies of order two, denoted as $T_{2}$ and $H_{2}$ respectively. The Tsallis entropy $T_2$ of the distribution $\{p_{|\phi\rangle}^{A}(i)\}$, given by $T_{2}(A;|\phi\rangle) = 1 - \sum_{i}(p_{|\phi\rangle}^{A}(i))^{2}$, is also referred to as the {\it linear entropy}. It is a measure of the {\it purity} of the density operator $\rho = \sum_{i} p_{|\phi\rangle}^{A}(i)\state{a_{i}}$ corresponding to the post-measurement state resulting from a measurement of $A$ on $|\phi\rangle$. The R\'enyi entropy $H_2$ (often called the {\it collision entropy}) of the distribution $\{p_{|\phi\rangle}^{A}(i)\}$ is given by $H_{2}(A; \ket{\phi}) = - \log\sum_{j=1}^d |\braket{\psi_{j}}{\phi}|^4$. 

In the following sections, we show that the incompatibility measures discussed above are indeed bounded below by the EUR lower bounds formulated in terms of the $T_{2}$ entropy.

\subsection{Entropic lower bounds on $\cQ$}\label{sec:accFid_EUR}

Evaluating the measure $\cQ$ requires a two-step optimization in general --  one over POVMs and the other over all possible state reconstruction maps. This problem is made tractable by first maximizing the average fidelity over all possible state reconstruction strategies. Following~\cite{Fuchs-Sasaki, fuchs04} the maximum fidelity function can thus be simplified as,
\begin{equation}
 F^{\rm max}_{\cS} = \frac{1}{d} \sup_{\cM \equiv \{M_{k}\}} \left[\sum_k \lambda_\mx \left(\mathbb{A}(M_{k})\right)\right], \label{eq:accFid}
\end{equation}
where, $\lambda_{\rm max}(.)$ is the maximum eigenvalue function, and, 
\[ \mathbb{A}(.) \equiv \frac{1}{N}\sum_{i,j} |a^{(i)}_{j}\rangle\langle a^{(i)}_{j}| (.)   | a^{(i)}_{j}\rangle\langle a^{(i)}_{j}| , \]
is the completely positive trace-preserving (CPTP) map whose Kraus operators are simply the states in the ensemble $\cS$. Furthermore, since the maximum is always attained for a POVM with rank-one elements, we can set $M_{k} = m_{k}|\chi_{k}\rangle\langle \chi_{k}|$ without loss of generality. The maximum fidelity for the ensemble $\cS$ is thus given by,
\begin{eqnarray}
F^{\rm max}_{\cS} &=& \frac{1}{d} \sup_{\cM \equiv \{m_{k}|\chi_{k}\rangle\}} \left[\sum_k m_{k} \lambda_\mx \left(\mathbb{A}(|\chi_{k}\rangle\langle\chi_{k}|)\right)\right] \nonumber \\
&=& \frac{1}{d} \sup_{\{m_{k},|\chi_{k}\rangle\}}\sum_{k}m_{k}F^{\rm avg}_{\cS}(|\chi_{k}\rangle\langle\chi_{k}|), \nonumber
\end{eqnarray}
where,
\begin{equation}
F^{\rm avg}_{\cS}(|\chi_{k}\rangle\langle\chi_{k}|) = \lambda_\mx \left(\mathbb{A}(|\chi_{k}\rangle\langle\chi_{k}|)\right) \label{eq:avgFid}
\end{equation}
is the {\it average fidelity} achieved by the POVM element $|\chi_{k}\rangle\langle\chi_{k}|$, for the ensemble $\cS$. 

A related notion which will be useful in stating our result, is that of a POVM that is {\it constant} with respect to a given ensemble of states $\cS$. 
\begin{definition}
A POVM $\cM \equiv \{m_{k}, |\chi_{k}\rangle\langle\chi_{k}|, k=1,\ldots, N\}$ with rank-one elements $|\chi_{k}\rangle\langle\chi_{k}|$ 
is said to be {\bf constant} with respect to a given ensemble $\cS \equiv \{|a^{(i)}_{j}\rangle\langle a^{(i)}_{j}|\}_{i,j}$ if the individual POVM elements all yield the same average fidelity $F^{\rm avg}_{\cS}(|\chi_{k}\rangle\langle\chi_{k}|)$ defined in Eq.~\eqref{eq:avgFid}. 
\end{definition} 
In other words, for a constant POVM, the average fidelity function $F^{\rm avg}_{\cS}(|\chi_{k}\rangle\langle\chi_{k}|)$ defined in Eq.~\eqref{eq:avgFid} is a constant for a given $\cS$, independent of the individual POVM elements $\{|\chi_{k}\rangle\langle \chi_{k}|\}$. 

We now prove the following relation between $\cQ$ and the $T_{2}$ entropic lower bound.
\begin{theorem}\label{thm:Q_EUR}
For a set of $N$ non-degenerate observables $\{A^{(1)}, A^{(2)}, \ldots, A^{(N)}\}$ with the associated ensemble of eigenstates $\cS \equiv \{\frac{1}{Nd},|\psi^{(i)}_{j}\rangle\langle\psi^{(i)}_{j}|\}$,  the incompatibility $\cQ(A^{(1)}, A^{(2)}, \ldots, A^{(N)})$ is bounded from below by the minimum average Tsallis $T_{2}$ entropy, that is,
\begin{equation}
\cQ(A^{(1)}, A^{(2)}, \ldots, A^{(N)}) \geq t_{2} (A^{(1)}, A^{(2)},\ldots, A^{(N)}), \label{eq:Q_EUR}
\end{equation}
where,
\[ t_{2}(A^{(1)}, A^{(2)}, \ldots, A^{(i)}) \equiv \min_{|\phi\rangle}\frac{1}{N}\sum_{i}T_{2}(A^{(i)};|\phi\rangle).
\]
Equality holds iff the optimal POVM achieving the accessible fidelity is constant with respect to the ensemble $\cS$.
\end{theorem}

\begin{proof}
As noted in Eq.~\eqref{eq:avgFid} above, the maximum fidelity for the ensemble $\cS$ is given by,
\begin{eqnarray*}
&& F^{\rm max}_{\cS} \nonumber \\
&=& \frac{1}{Nd} \sup_{\{m_k, \ket{\chi_k}\}} \sum_{k=1}^{d^2} m_k \lambda_\mx\left[\sum_{i=1, j=1}^{N,d}|\braket{\chi_k}{a^i_j}|^2\state{a^i_j}\right] \\
&=& \frac{1}{Nd} \sup_{\{m_k, \ket{\chi_k}\}} \sum_{k=1}^{d^2} m_k \max_{\ket{\gamma} = 1}\;\sum_{i=1, j=1}^{N,d}|\braket{\chi_k}{a^i_j}|^2|\braket{\gamma}{a^i_j}|^2,
\end{eqnarray*}
where the supremum is taken over all rank-one POVMs $\{m_{k}, \ket{\chi_{k}}\}$, satisfying $\sum_k m_k \state{\chi_k} = \Id$. To evaluate the maximization over $\ket{\gamma}$, notice that the sum over $i,j$ is of the form
\[ \sum_s q_s^2 r_s^2 \quad \text{with} \quad \sum_s q_s^2 = \sum_s r_s^2 = N. \]
For fixed $q_s^2 = |\braket{\phi_k}{a^i_j}|^2$, the quantity is maximized when $r_s^2$ is proportional to $q_s^2$, by Cauchy-Schwarz inequality. Now since the sums of $q_s^2$ and $r_s^2$ are equal, this means that the maximizing $\ket{\gamma}$ has $|\braket{\gamma}{a^i_j}|^2 = |\braket{\phi_k}{a^i_j}|^2$. Therefore,
\begin{align}
F^{\rm max}_{\cS} & = \frac{1}{Nd} \sup_{\{m_k, \state{\chi_k}\}} \sum_{k=1}^{d^2} m_k \sum_{i=1, j=1}^{N,d}|\braket{\chi_k}{a^i_j}|^4 \nonumber \\
                 & \le \frac{1}{N} \max_{\ket{\phi}} \sum_{i=1}^{N}\sum_{j=1}^{d}|\braket{\phi}{a^i_j}|^4 .  \label{eq:Facc_ub}
\end{align}

Thus we have,
\begin{eqnarray}
\cQ(A^{(1)},\ldots,A^{(N)}) &\geq& 1 - \frac{1}{N}\max_{\ket{\phi}} \sum_{i=1}^{N}\sum_{ j=1}^{d}|\braket{\phi}{\psi^i_j}|^4 \nonumber \\
&= & \min_{\ket{\phi}}\frac{1}{N}\sum_{i=1}^{N}\left(1  - \sum_{j=1}^{d}|\braket{\phi}{\psi^i_j}|^4 \right) \nonumber \\
&=& \min_{\ket{\phi}} \frac{1}{N}\sum_{i=1}^{N}T_{2}(A^{(i)};|\phi\rangle),
\end{eqnarray}
as desired. Clearly, equality holds iff the inequality in Eq.~\eqref{eq:Facc_ub} is saturated. This holds iff the optimal POVM achieving the maximum fidelity $F^{\rm max}_{\cS}$ is constant, that is, every element of the optimal POVM achieves the same average fidelity for the ensemble $\cS$. 
\end{proof}

Furthermore, from the definition of the $H_{2}$ entropy and the upper bound in Eq.~\eqref{eq:Facc_ub}, it follows that,
\[ \cQ(A^{(1)}, A^{(2)}, \ldots, A^{(N)}) \geq 1 - \max_{\ket{\phi}} \frac{1}{N} \sum_{i=1}^N 2^{-H_{2}(A^{(i)}, \ket{\phi})}.\]
Applying the inequality of arithmetic and geometric means, we have,
\[ \frac{1}{N} \sum_{i=1}^{N} 2^{-H_{2}(A^{(i)}, \ket{\phi})} \geq 2^{-\frac{1}{N}\sum_{i=1}^{N}H_{2}(A^{(i)}, \ket{\phi})},\]
where equality holds iff the individual entropies corresponding to a given state $|\phi\rangle$ are all equal. Thus, we have the following interesting relation between the incompatibility $\cQ(A^{(1)}, A^{(2)}, \ldots, A^{(N)})$ and the minimum average $H_{2}$ entropy of the set of observables $\{A^{(1)},\ldots, A^{(N)}\}$.
\begin{corollary}
 Consider $N$ non-degenerate observables $\{A^{(1)}, A^{(2)}, \ldots, A^{(N)}\}$ satisfying a tight $H_{2}$ EUR of the form
\[ \min_{|\phi\rangle}\frac{1}{N}\sum_{i=1}^{N}H_{2}(A^{(i)};|\phi\rangle) = c_{2}(A^{(1)}, A^{(2)},\dots,A^{(N)}),\]
such that, the observables $A^{(i)}$ all have equal entropy on the minimizing state $|\phi\rangle$. The mutual incompatibility of such a set of $N$ non-degenerate observables $\{A^{(1)}, A^{(2)}, \ldots, A^{(N)}\}$ satisfies,
\begin{equation}
\cQ(A^{(1)}, A^{(2)}, \ldots, A^{(N)}) \geq 1 - 2^{-c_{2}(A^{(1)}, A^{(2)},\ldots, A^{(i)})}. \label{eq:Q_H2EUR}
\end{equation}
\end{corollary}

One important class of observables for which the inequalities derived above are saturated, are the mutually unbiased bases. Recall that two orthonormal bases $\{|a_{i}\rangle\}$, $\{|b_{i}\rangle\}$ in $\cH_{d}$ are said to be mutually unbiased iff $|\langle a_{i}|b_{j}\rangle|^{2} = \frac{1}{d}$, for all $i,j$. In~\cite{bandyoPM_PRA13}, it was shown that the incompatibility of a set of mutually unbiased bases (MUBs) $\{\cB^{(1)}, \cB^{(2)}, \ldots, \cB^{(N)} \} \in \cH_{d}$ is given by
\begin{equation}
 \cQ(\cB^{(1)}, \cB^{(2)}, \ldots, \cB^{(N)})  = \left(1 - \frac{1}{N}\right)\left(1 - \frac{1}{d}\right) .  \label{eq:Q_MUB}
 \end{equation}
This coincides with the lower bound on the average Tsallis($T_2$) entropy of a set of $N$ MUBs in $d$-dimensions~\cite{wu:MUB_bound, PM_MDSdisturbance}
 \[ \inf_{\ket{\phi}} \frac{1}{N} \sum_{i=1}^N T_{2}(\cB^{(i)}; \ket{\phi}) \geq \left(1 - \frac{1}{N}\right)\left(1 - \frac{1}{d}\right) , \]
thus showing that Eq.~\eqref{eq:Q_EUR} is indeed tight for the case of mutually unbiased bases. The optimal POVM that achieves the accessible fidelity in this case is one of the bases $\cB^{(i)}$ in the set, and the mutual unbiasedness property naturally ensures that it is {\it symmetric} for the corresponding ensemble of eigenstates. Similarly, by comparing with the well-known $H_{2}$ EUR for a set of $N$ MUBs in a $d$-dimensional space~\cite{ww:urSurvey},
\begin{align*}
 \inf_{\ket{\phi}} \frac{1}{N} \sum_{i=1}^{N} H_{2}(\cB^{(i)}; \ket{\phi}) & \geq -\log\frac{N + d -1}{Nd} \nonumber \\
 &  \equiv c_{2}(\cB^{(1)},\ldots, \cB^{(N)}) \nonumber \\
 &  = - \log\left[1 -\cQ(\cB^{(1)}, \cB^{(2)}, \ldots, \cB^{(N)})\right], \nonumber\\
\end{align*}
we see that Eq.~\eqref{eq:Q_H2EUR} is also saturated in this case.

\subsection{Entropic lower bounds on $\cQ_{F}$}\label{sec:QF_EUR}

We now show that a similar relation holds between the fidelity-based incompatibility measure $\cQ_{F}$ and the Tsallis ($T_{2}$) EURs formulated for a successive measurement scenario. Let $t_{2}(A\rightarrow B)$ and $t_{2}(B\rightarrow A)$ denote the entropic lower bounds corresponding to observables $A$ and $B$, as defined below:
\begin{eqnarray}
\frac{1}{2}\inf_{\ket{\psi}} \left[ T_{2}(A; \ket{\psi}) + T_{2} (B; \cE^{A}(\ket{\psi}))\right] &=& t_{2} (A\rightarrow B), 
\nonumber  \\
\frac{1}{2}\inf_{\ket{\psi}} \left[ T_{2}(B; \ket{\psi}) + T_{2} (A; \cE^{B}(\ket{\psi}))\right] &=& t_{2} (B\rightarrow A). \nonumber
\end{eqnarray}
Recall from Lemma~\ref{lem:succEUR} that,
\begin{eqnarray}
t_{2}(A\rightarrow B) &=& \inf_{k}\left\{ 1-\sum_{j}\left|\left\langle a_{k}\right|\left.b_{j}\right\rangle \right|^{4}\right\} \\ \nonumber
t_{2} (B\rightarrow A) &=& \inf_{\ket{b_{j}}}\left\{ 1-\sum_{k}\left|\left\langle a_{k}\right|\left.b_{j}\right\rangle \right|^{4}\right\}. \label{eq:uncertB_A}
\end{eqnarray}
In the following theorem we prove that the average entropic lower bound 
\[ t_{2}^{\rm succ}(A,B) = 1/2(t_{2}(A\rightarrow B) + t_{2}(B\rightarrow B))\] constitutes a lower bound for the incompatibility measure $\cQ_{F}(A,B)$.
\begin{theorem}\label{thm:QF_EUR}
For a pair of observables $A, B$, the measure $\cQ_{F}(A,B)$ satisfies,
\begin{equation}
\cQ_{F}(A,B) \geq  t_{2}^{\rm succ}(A,B). \label{eq:QF_EUR}
\end{equation}
\end{theorem}
\begin{proof}
The result follows from the observation that the EUR in the successive measurement scenario for a measurement of $ A$ followed by $B$ is related to the fidelity of the statistics of a measurement of $ B$ followed by $A$ and vice versa. From the definition of $\cQ_{F}(A\rightarrow B)$ we have,  
\begin{align*}
& Q_{F}(A\rightarrow B) \\
& =\sup_{\ket \psi}\left\{ 1-\left[\sum_{j}\sqrt{\left(\left|\left\langle \psi\right.\left|b_{j}\right\rangle \right|^{2}\right)\sum_{i}\left|\left\langle a_{i}\right.\left|b_{j}\right\rangle \right|^{2}\left|\left\langle \psi\right.\left|a_{i}\right\rangle \right|^{2}}\right]^{2}\right\} \\
 & \geq\sup_{k}\left\{ 1-\sum_{i}\left|\left\langle a_{i}\right.\left|b_{k}\right\rangle \right|^{4}\right\},
 \end{align*}
where $\ket{b_{k}}$ is an eigenket of $B$.  Then Eq.~\eqref{eq:uncertB_A} implies,
\begin{align*}
\cQ_{F}(A\rightarrow B) & \geq\inf_{k}\left\{ 1-\sum_{i}\left|\left\langle a_{i}\right.\left|b_{k}\right\rangle \right|^{4}\right\} \\
 & =2t_{2}(B\rightarrow A) .
\end{align*}
Similarly, we can show $Q_{F}(B\rightarrow A)\geq 2t_{2}(A\rightarrow B)$. Therefore, in terms of $Q_{F}(A,B)$, we have,
\begin{eqnarray}
Q_{F}(A,B) &=& \frac{1}{4}\left[Q_{F}(A\rightarrow B)+Q_{F}(B\rightarrow A)\right] \nonumber \\
&\geq & \frac{1}{2}\left[t_{2}(A\rightarrow B)+ t_{2}(B\rightarrow A)\right] \nonumber \\
&=& t_{2}^{\rm succ}(A,B).
\end{eqnarray}
\end{proof} 


Equality is attained iff 
\[ \underset{\left|b_{k}\right\rangle }{\sup}\sum_{i}\left|\left\langle a_{i}\right.\left|b_{k}\right\rangle \right|^{4}=\underset{\left|b_{k}\right\rangle }{\inf}\sum_{i}\left|\left\langle a_{i}\right.\left|b_{k}\right\rangle \right|^{4} .\]
In other words, equality holds iff the quantity $\sum_{i}\left|\left\langle a_{i}\right.\left|b_{k}\right\rangle \right|^{4}$ is a constant for all eigenkets $|b_{k}\rangle$ of observable $B$. 

It is easy to see that the condition for equality is satisfied for the case of mutually unbiased observables. Indeed, for a pair of mutually unbiased bases $\cB^{(1)}, \cB^{(2)}$ in a $d$-dimensional space, it was shown that~\cite{PM_MDS14},
\[ \cQ_{F}(\cB^{(1)}, \cB^{(2)}) = \frac{1}{2}\left(1 - \frac{1}{d}\right).\]
The $T_{2}$ entropic lower bound may be evaluated as follows: for a measurement of $\cB^{(1)}$ followed by $\cB^{(2)}$  on the same system,  
\begin{eqnarray}
t_{2}(\cB^{(1)}\rightarrow \cB^{(2)}) &=& \inf_{\ket{\psi}}\frac{1}{2}\left[ T_{2}(\cB^{(1)}; |\psi\rangle)+T_{2}\left(\cB^{(2)};\cE^{\cB^{(1)}}(|\psi\rangle\langle\psi|)\right)\right]  \nonumber \\
&=& \inf_{\ket{\psi}}\frac{1}{2}\left(1-\sum_{j=1}^{d}\left|\left\langle \psi\right|\left.b_{j}\right\rangle \right|^{4}\right) \nonumber \\
 &=& \frac{1}{2}\left(1-\frac{1}{d}\right). \nonumber
\end{eqnarray}
Since the lower bound is independent of the order in which the MUBs are measured, $t_{2}(\cB^{(1)}\rightarrow \cB^{(2)}) = t_{2}(\cB^{(2)}\rightarrow \cB^{(1)})$. Therefore,
\[ t_{2}(\cB^{(1)}, \cB^{(2)}) = \frac{1}{2}\left(1-\frac{1}{d}\right),\]
showing that Eq.~\eqref{eq:QF_EUR} is indeed tight for a pair of MUBs.





\subsection{Incompatibilty of qubit observables}\label{sec:Q_qubit}

Here we obtain an exact expression for the incompatibility of a pair of qubit observables, and show that the bounds obtained above are saturated in this case. While evaluating $\cQ(A,B)$ and $\cQ_{F}(A,B)$ involves solving a hard optimization problem in general, the problem can be simplified for the qubit case by making use of the Bloch sphere representation. Thus, we parametrize $A$ and $B$ in terms of unit vectors $\vec{a}, \vec{b} \in \mathbb{R}^{3}$ as follows: $A \equiv \left\{\frac{1}{2}\left(\Id \pm \vec{a}.\vec{\sigma}\right)\right\}$ and $B \equiv \left\{\frac{1}{2}\left(\Id \pm \vec{b}.\vec{\sigma}\right)\right\}$, where, $\vec{\sigma} = (\sigma_{X}, \sigma_{Y}, \sigma_{Z})$ denote the Pauli matrices and $\Id$ denotes the $2\times 2$ identity matrix. 

We merely state the result here and refer to Appendix~\ref{sec:Q_qubitProof} for the proof.
\begin{theorem}\label{thm:Q_qubit}
Given a pair of qubit observables $A, B$ parameterized by real vectors $\vec{a}, \vec{b} \in \mathbb{R}^{3}$ respectively, their mutual incompatibilities are given by
\begin{eqnarray}
\cQ(A,B) &=& \frac{1}{4}\left(1 - |\vec{a}.\vec{b}| \right),  \\
Q_{F}(A, B)&=& \frac{1}{4}\left(1-(\vec{a}.\vec{b})^{2}\right) .
\end{eqnarray}
\end{theorem}

We see that the measures $\cQ(A,B)$ and $\cQ_{F}(A,B)$ coincide for a pair of MUBs ($\vec{a}.\vec{b} = 0$) and for the case when $A$ and $B$ commute ($\vec{a}.\vec{b} =1$). For any other pair of qubit observables ($0 < \vec{a}.\vec{b} < 1$) we have $\cQ_{F}(A,B) > \cQ(A,B)$; the mutual incompatibility measure formulated in the successive measurement scenario is higher (see Fig.~\ref{fig:Q,QF}).

We further note that $\cQ(A,B)$ coincides with the recently obtained lower bound on the average $T_{2}$ entropy of a pair of qubit observables~\cite{PM_MDSdisturbance}:
\begin{eqnarray}
&& \min_{\ket{\phi}} \frac{1}{2}\left[T_{2}(A;|\phi\rangle) + T_{2}(B;|\phi\rangle)\right] \nonumber \\
&=& \frac{1}{4}(1-|\vec{a}.\vec{b}|) =  \cQ(A,B). \nonumber
\end{eqnarray}
The corresponding $H_{2}$ entropies satisfy~\cite{BPP12_H2entropy},
\begin{eqnarray}
&& \min_{\ket{\phi}} \frac{1}{2}\left[H_2(A; \ket{\phi}) + H_2(B; \ket{\phi})\right] \nonumber \\
&=& - \log\left(\frac{3}{4} + \frac{1}{4}|\vec{a}\cdot\vec{b}|\right) = -\log[1- \cQ(A,B)], \nonumber
\end{eqnarray}
thus showing that the bounds obtained in Theorem~\ref{thm:Q_EUR} are tight for the case of qubit observables.

In the successive measurement scenario, the $T_{2}$ EUR for qubit observables $A, B$ can be evaluated as follows:
\begin{eqnarray}
&& \inf_{\ket{\psi}}\frac{1}{2}\left[ T_{2}\left(A;\cE^{A}(\left|\psi\right\rangle \left\langle \psi\right|)\right)+T_{2}\left(B;\cE^{A}(\left|\psi\right\rangle \left\langle \psi\right|)\right)\right] \nonumber \\
&=& \inf_{i}{\frac{1}{2}}T_{2}\left({\rm Pr}_{\left|a_{i}\right\rangle }^{A\rightarrow B}(j)\right) \nonumber \\
&=& \frac{1}{4}\left(1-\left(\vec{a}\cdot\vec{b}\right)^{2}\right) = \frac{1}{2}Q_{F}(A\rightarrow B), \nonumber
\end{eqnarray}
thus showing that the bound in Theorem~\ref{thm:QF_EUR} is also tight for qubit observables.

\begin{figure}
\includegraphics[scale=0.64]{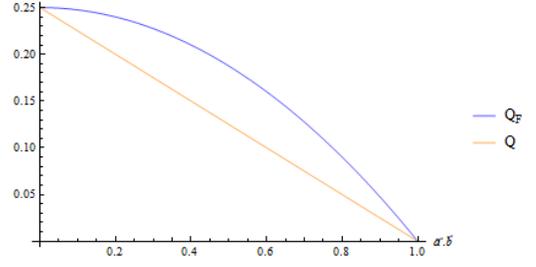}
\caption{Comparing $\cQ$ and $\cQ_{F}$ for qubit observables with Bloch vectors $\vec{a}, \vec{b}$, as a function of the angle $\vec{a}.\vec{b}$.}\label{fig:Q,QF}
\end{figure}

\subsection{Observables that commute on a subspace}\label{sec:commute}

Finally, we discuss a class of observables for which the measures $\cQ$ and $\cQ_{F}$ are strictly better than the EUR formalism in quantifying incompatibility, namely, observables that commute on a subspace. An example of such a set of observables comes from the theory of angular momentum, where the operators $L_{x}$ and $L_{z}$ do not commute but still have the $l=0$ state as a common eigenvector~\cite{sakurai}. The commuting subspace is one-dimensional in this case.

We consider a pair of non-degenerate observables $A,B$ that commute over a subspace of dimension $d_{c}$. We further assume that they are mutually unbiased in the $(d-d_{c})$-dimensional subspace on which they do not commute. The eigenstates $\{\ket{a_i}\}_{1=1}^{d}$ and $\{\ket{b_i}\}_{1=1}^{d}$ of $A$ and $B$ therefore satisfy,
\begin{equation}
|\braket{a_i}{b_j}| = \begin{cases} \delta_{i,j} \text{ for } i \text{ or } j \in \{1, \dots, d_c\} \\ \frac{1}{\sqrt{d- d_c}} \text{ for } i \text{ and } j \in \{d_c+1, \dots, d\} \end{cases}\label{eq:obs_comm}
\end{equation}
As observed above, the lower bound $c(A,B)$ (defined in Eq.~\eqref{eq:eur}) on the average entropies for such a pair of observables is indeed zero, with the minimizing state being any of the common eigenstates $\ket{a_i} \equiv \ket{b_i}, \; i \in [d_c]$. That is,
\[ \inf_{|\phi\rangle} \frac{1}{2}\left( H(A;|\phi\rangle) + H(B;|\phi\rangle) \right) = 0, \]
for any entropic quantity $H$. The measures $\cQ(A,B), \cQ_{F}(A,B)$, on the other hand, gives a non-trivial value for the incompatibility of such a pair of observables.
\begin{theorem}
Consider a pair of observables $A,B$ in $\cH_{d}$ which commute on a subspace of dimension $d_{c} < d$ and are mutually unbiased in the $(d-d_{c})$ subspace. The mutual incompatibility of such a pair is given by
\begin{eqnarray}
\cQ (A,B) &=& \frac{1}{2}\left(1 - \frac{d_{c} + 1}{d} \right) , \label{eq:Q_dc} \\
\cQ_{F}(A,B) &=& \frac{1}{2}\left(1 - \frac{1}{d-d_{c}} \right) . \label{eq:QF_dc}
\end{eqnarray}
\end{theorem}
As in the case of qubit observables, here again the incompatibility measure $\cQ_{F}$ is larger than the measure $\cQ$:
\begin{eqnarray}
\cQ_{F}(A,B) &=& \frac{1}{2}\left(1 - \frac{1}{d-d_{c}}\right) \nonumber \\
&=& \frac{1}{2}\left(\frac{d-d_{c}-1}{d-d_{c}}\right) \nonumber \\
&\geq& \frac{1}{2}\left(\frac{d-d_{c}-1}{d}\right) = \cQ(A,B). \nonumber
\end{eqnarray}
The two measures match for $d_{c}=0$, in which case $A$ and $B$ are mutually unbiased, and for $d_{c}=d-1$, in which case $A$ and $B$ commute (see Fig.~\ref{fig:dc}).
\begin{figure}
\includegraphics[scale=0.5]{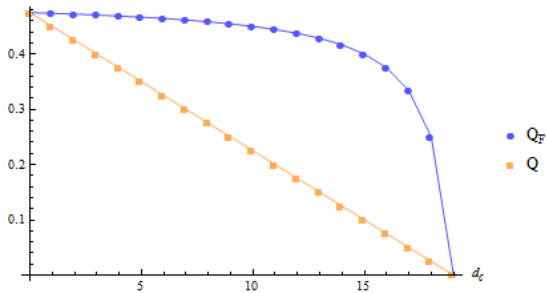}
\caption{Comparing $\cQ$ and $\cQ_{F}$ for observables that commute on a subspace of dimension $d_{c}$, in a space of dimension $d=20$.}\label{fig:dc}
\end{figure} 
\begin{proof}
Eq.~\eqref{eq:QF_dc} for the fidelity-based incompatibility measure was originally evaluated in~\cite{PM_MDS14}. Here, we prove Eq.~\eqref{eq:Q_dc} for the measure $\cQ(A,B)$ as follows. 

Consider the ensemble of eigenstates of the observables $A,B$ defined in Eq.~\eqref{eq:obs_comm}, namely, $\cS\equiv \{|a_{i}\rangle\langle a_{i}|, |b_{i}\rangle\langle b_{i}|\}_{i=1}^{d}$. $\cS$ can be written as a direct sum of two ensembles: $\cS \equiv \cS_{1}\oplus\cS_{2}$, where, $\cS_{1}$ comprises the common eigenstates in the commuting subspace, and $\cS_{2}$ comprises the unbiased states in the non-commuting subspace. We now make use of the following useful property of the maximum fidelity function, namely, that $F^{\rm max}_{\cS}$ is additive for ensembles belonging to distinct subspaces. We merely state the result here and present a proof in Appendix~\ref{sec:additivity_proof}.
\begin{lemma} \label{lem:FAdditivity}
Given ensembles $\cS_1 \in \cH_{d_1}$ and $\cS_2 \in \cH_{d_2}$, consisting of $N_1d_1$ and $N_2d_2$ states respectively, the maximum fidelity of the ensemble $\cS_1 \oplus \cS_2 \in \cH_{d_1}\oplus\cH_{d_2}$ is given by
\[ F^{\rm max}_{\cS_1 \oplus \cS_2} = \frac{1}{N_1 d_1 + N_2d_2}\left(N_1d_1 F^{\rm max}_{\cS_1} + N_2d_2 F^{\rm max}_{\cS_2}\right) . \]
\end{lemma}

In this case, since $\cS_{1}$ is an ensemble of orthogonal states, $F^{\rm max}_{\cS_{1}} = 1$. For $\cS_{2}$, we can simply invoke the result for MUBs from~\cite{bandyoPM_PRA13}, to get, 
\[ F^{\rm max}_{\cS_{2}} = \frac{d-d_{c} +1}{2(d-d_{c})} .\]
Invoking the additivity result above, we get,
\begin{eqnarray*}
\cQ(A,B) &=& 1 - F^{\rm max}_{\cS} \nonumber \\
&=& 1 - \frac{1}{d}\left(d_c\cdot 1 + (d -d_{c}) \cdot \frac{d - d_{c} + 1}{2(d-d_{c})}\right) \\
&=& \frac{1}{2}\left(1 - \frac{d_{c} + 1}{d} \right). \nonumber
\end{eqnarray*}
 \end{proof}

\section{Evaluating $\cQ, \cQ_{F}$ for a general set of observables}\label{sec:Q_lbound}

Obtaining an exact expression for the mutual incompatibilities $\cQ, \cQ_{F}$ of an arbitrary set of observables involves solving a hard optimization problem in general. It is known that the mutual incompatibility of a set of $N$ MUBs in $\cH_{d}$ constitutes an upper bound for the incompatibility of any set of $N$ observables in $\cH_{d}$~\cite{bandyoPM_PRA13, PM_MDS14}. Here, we obtain non-trivial lower bounds on the mutual incompatibility of a general set of observables, and prove that these lower bounds are in fact efficiently computable.

\subsection{An efficiently computable lower bound for $\cQ$ }\label{sec:matrixNorm}

As defined in Sec.~\ref{sec:F_acc}, the incompatibility $\cQ(A^{(1)}, A^{(2)}, \ldots, A^{(N)})$ of a set of $N$ observables is related to the maximum fidelity $F^{\rm max}_{\cS}$ attainable for the ensemble $\cS$ of their eigenstates. It is known that computing $F^{\rm max}_{\cS}$ for a quantum state estimation process for a general ensemble of states $\cS$ involves a sequence of semidefinite programs (SDPs)~\cite{navascues}. 
An SDP is an efficiently computable convex optimization problem of the general form~\cite{BoydVan}:
\[ \min\langle C, X\rangle \quad \text{subject to} \quad \Phi(X) = A, \quad X \succeq 0 \]
where $C$, $A$ are matrices and $\Phi$ is a linear operator. For a given ensemble of states $\cS$, the algorithm in~\cite{navascues} for computing $F^{\rm max}_{\cS}$ only guarantees asymptotic convergence: while each bound in the sequence may be efficiently computable, the number of steps required for the sequence to converge can be quite large.

Here, we show that by recasting the maximum fidelity function as a matrix norm, we obtain an efficiently computable lower bound for the incompatibility $\cQ$ of any set of observables. We first note that the two-fold optimization in $F^{\rm max}_{\cS}$ can be recast as a single optimization over a class of CPTP maps as follows.
\begin{eqnarray}
 F^{\rm max}_{\cS} & = & \frac{1}{Nd}\sup_{\cM, \cR} \sum_{k,i,j}\langle a_{j}^{(i)}| M_{k}| a^{(i)}_{j}\rangle \langle a_{j}^{(i)}| \sigma_{k}| a^{(i)}_{j}\rangle \nonumber \\
&=& \frac{1}{Nd} \max_{\mathbb{M}}\sum_{i,j} \matel{a^{i}_{j}}{\mathbb{M}(\state{a^i_j})}{a^i_j}, \nonumber
\end{eqnarray}
where the map $\mathbb{M}$ is the entanglement-breaking channel~\cite{HSR03} corresponding to the positive operators $\{M_{k}\}$ and the states $\{\sigma_{k}\}$:
\[ \mathbb{M}(\rho) = \sum_k \Tr(M_k \rho) \sigma_k. \]
It was shown that~\cite{CX13} this optimization over the class of entanglement-breaking channels can be further rewritten as a minimum norm of an operator. Specifically, 
\[
F^{\rm max}_{\cS} = \frac{1}{Nd}\min_{\substack{\rho: \rho \succ 0 \\ \Tr(\rho) = 1}} \left|\left| \sum_{i,j} \state{a^i_j} \otimes \rho^{-1/2}\state{a^i_j}\rho^{-1/2}\right|\right|_\times,
\]
where, $\parallel (.) \parallel_{\times}$ denotes the injective cross norm, defined as \[ ||M||_\times = \sup_{||\alpha|| = ||\beta|| = 1} \bra{\alpha}\bra{\beta} M \ket{\beta}\ket{\alpha}. \]
While the injective norm itself is not efficiently computable in general, it is bounded from above by the standard operator norm $||.||_{\infty}$ given by $||M||_{\infty} = \sup_{\parallel \alpha\parallel = 1} \langle\alpha|M|\alpha\rangle$. Therefore,
\[ F^{\rm max}_{\cS} \leq \frac{1}{Nd}\min_{\substack{\rho: \rho \succ 0 \\ \Tr(\rho) = 1}} \left|\left| \sum_{i,j} \state{a^i_j} \otimes \rho^{-1/2}\state{a^i_j}\rho^{-1/2}\right|\right|_\infty\]
The optimization on the right hand side can be further simplified as~\cite{CX13}
\[
\min_{\substack{\Lambda: \Lambda \succ 0 \\ A \preceq \Id \otimes \Lambda}} \Tr[\Lambda], \quad A = \frac{1}{Nd}\state{a^i_j}\otimes\state{a^i_j}.
\]
This is basically the problem of minimizing a maximum eigenvalue which is easily seen to be an SDP, and efficient algorithms for solving this are known~\cite{BoydVan}.

\subsection{Computability of $Q_{F}$}

Here we prove a non-trivial lower bound on the fidelity-based incompatibility measure $\cQ_{F}(A,B)$ of any pair of observables, which can be recast as a  convex program. Recall from Sec.~\ref{sec:incompat_dist} that the measure $\cQ_{F}$ is defined as a supremum over all states $|\psi\rangle$ of the fidelity between the probability distributions $\{{\rm Pr}_{\ket{\psi}}^{B}(j)\}$ and $\{{\rm Pr}_{\ket{\psi}}^{A\rightarrow B}(j)\}$ defined in Eq.~\eqref{eq:prob_defn}. We first note that the fidelity $F^{2}[P,Q]$ for a pair of probability distributions $P \sim \left\{p_{i}\right\}$ and $Q \sim \left\{q_{i}\right\}$ can be bounded as follows,
\begin{align*}
F^{2}\left[ P,Q\right]&=\left[\sum\limits_{i=1}^{d}\sqrt{p_{i}}\sqrt{q_{i}}\right]^2\\
&= 2\sum\limits_{i=1}^{d}p_{i}q_{i}, \nonumber
\end{align*}
by using the arithmetic mean to bound the geometric mean. Thus the measure $\cQ_{F}$ is bounded by,


\begin{align*}
& Q_{F}(A\rightarrow B) = 1-F^{2}\left[{\rm Pr}_{\ket{\psi}}^{B}, {\rm Pr}_{\ket{\psi}}^{A\rightarrow B} \right]\\
 &= 1-\inf_{\ket{\psi}}\left[\sum\limits_{j=1}^{d}\sqrt{ \left| \braket{\psi}{b_{j}} \right| ^{2}}\sqrt{\sum\limits_{k=1}^{d} \left| \braket{\psi}{a_{k}}\right|^{2}\, \left| \braket{a_{k}}{b_{j}}\right|^{2}}\,\right]^{2}\\
 &\geq 1-2\inf_{\ket{\psi}}\sum\limits_{j=1}^{d} \left| \braket{\psi}{b_{j}} \right| ^{2}\left[\sum\limits_{k=1}^{d} \left| \braket{\psi}{a_{k}}\right|^{2}\, \left| \braket{a_{k}}{b_{j}}\right|^{2}\right] \nonumber \\
 &= 1- g(A,B),\nonumber
\end{align*}
where we have defined
\[g(A,B) := 2\inf_{\ket{\psi}}\sum\limits_{j=1}^{d} \left| \braket{\psi}{b_{j}} \right| ^{2}\left[\sum\limits_{k=1}^{d} \left| \braket{\psi}{a_{k}}\right|^{2}\, \left| \braket{a_{k}}{b_{j}}\right|^{2}\right]. \]
Let $\alpha_{n} = \langle a_{n}|\psi\rangle$ and $\beta_{m} = \langle b_{m} |\psi\rangle$ denote the overlap coefficients of $|\psi\rangle$ with the eigenstates of $A$ and $B$ respectively. The function $g(A,B)$ can then be bounded as follows.
\begin{eqnarray}
g(A,B) &=& \min \sum_{i,j=1}^{d} \left|\alpha_{i}\right|^{2} \, \left|\beta_{j}\right|^{2} \, \left| \braket{a_{i}}{b_{j}}\right|^{2} \nonumber \\
&\leq& \min_{\{\alpha_{i}\}} \sum_{j}^{d} \sum\limits_{k}^{d} \left|\alpha_{k}\right|^{2} \, \left|\braket{b_{j}}{a_{k}}\right|^{2} \sum\limits_{i}^{d} \left|\alpha_{i}\right|^{2} \, \left| \braket{a_{i}}{b_{j}}\right|^{2} \nonumber \\
&=& \min_{\{ \alpha_{i}\}} \sum_{i,k}^{d} \left|\alpha_{i}\right|^{2} \left(\sum_{j} \left|\braket{b_{j}}{a_{k}}\right|^{2} \left| \braket{a_{i}}{b_{j}}\right|^{2}  \right) \left|\alpha_{k}\right|^{2} \nonumber \\
&=& \min_{\substack{\sum_{i}v_{i} = 1 \\ v \succeq 0}} v^{T} A v,
\end{eqnarray}
where $v$ is the matrix with elements $v_{i}=\left|\alpha_{i}\right|^{2}$ and ${\emph A}$ is the constant matrix with elements $A_{ij}=\left| \braket{a_{i}}{b_{j}}\right|^{2}$.  

In other words, the solution to the following optimization problem
\begin{equation*}
\begin{aligned}
& {\text{minimize:}}
& & v^{T} A v \\
& \text{subject to:}
& & \sum_{i} v_{i} = 1\\
&&& v \succeq 0.
\end{aligned}
\end{equation*}
gives a non-trivial lower bound for the incompatibility $\cQ_{F}(A,B)$ for an arbitrary pair of observables. This minimization problem is indeed in the form of a convex program~\cite{BoydVan}, it is therefore efficiently computable with standard convex optimization routines.

\section{Concluding Remarks}

In summary, our work offers a comparative study of the different approaches used to quantify the mutual incompatibility of quantum observables. We consider two recently proposed measures of incompatibility, namely, the measure $\cQ$ that is related to the maximum fidelity function in a state-discrimination context (as also the accessible fidelity in QKD), and the measure $\cQ_{F}$ that arises naturally in a successive measurement scenario. We show that these operational measures are lower bounded by the standard entropic uncertainty bounds formulated in terms of the Tsallis $T_{2}$ entropy. We also obtain conditions under which these incompatibility measures coincide exactly with the $T_{2}$ entropic lower bounds, and show that these conditions are satisfied for the case of MUBs and for a pair of qubit observables. 

We also consider the case of observables that commute on a subspace, which serves to highlight the fact that the measures $\cQ$ and $\cQ_{F}$ go beyond EURs in quantifying incompatibility. We obtain an exact expression for the incompatibilities $\cQ, \cQ_{F}$ of such a pair of observables, whereas the entropic uncertainty lower bound is simply zero. We further note the interesting fact that $\cQ_{F} \geq \cQ$ in this case, as well as for qubit observables. This highlights the fact that the mutual incompatibility of quantum observables manifests itself more strongly in a successive measurement scenario, a fact that is observed even in the corresponding entropic uncertainty bounds. 

While the incompatibility measures studied here are hard to compute in general, we do have exact expressions for the incompatibility of a pair of qubit observables. For a general set of observables, we obtain non-trivial lower bounds on their mutual incompatibility, which can be efficiently computed via convex optimization routines. Our bounds thus provide a useful tool for estimating the mutual incompatibility of an arbitrary set of observables.




\appendix

\section{Role of the measure $\cQ$ in QKD }\label{sec:Q_QKD}

In a generic ``prepare and measure'' QKD protocol Alice transmits a set of pure states drawn from different incompatible bases to Bob, who measures the received state in a basis of his choice. At the end of the protocol Alice and Bob compare their choice of bases, and retain only those states for which Bob's measurement basis coincides with the basis that Alice used. The corresponding outcomes represent the {\it raw key} in Bob's possession.

In the absence of errors, the raw key is already the secret key, and the eavesdropper has no information. However, in any practical protocol, Alice and Bob must account for the errors caused by the eavesdropper's presence and the secret key is obtained after correcting for these errors. Estimating the error rate is thus an important step in arriving at the final length of secret key extracted. 

\begin{lemma}
For a QKD protocol whose signal states are drawn uniformly at random from the eigenstate ensemble $\cS$, the measure $\cQ(A^{(1)},\ldots,A^{(N)})$ is the attainable lower bound on the error rate caused by an eavesdropper adopting an intercept-resend strategy.
\end{lemma}
The measure $\cQ$ for a set of observables is thus a benchmark for a QKD protocol whose signal states are drawn from the eigenstates of the given set, assuming the eavesdropper adopts an intercept-resend strategy.
\begin{proof}
Given the eavesdropper's choice of POVM $\cM$ and reconstruction map $\cR$, the final ensemble seen by Bob is $\cS' \equiv \{p_{a}(i,j), \sigma_{a}(i,j)\}$, where the probabilities $p_{a}$ are given by,
\[ p_{a}(i,j) = \tr[M_{a}|\psi_{j}^{(i)}\rangle\langle\psi_{j}^{(i)}|]. \] 
Assuming that Bob's choice of basis coincides with that of Alice, the average success probability for Bob to obtain the original state sent by Alice is given by, \begin{eqnarray} 
&& p_{\rm succ} \nonumber \\
 &=& \frac{1}{Nd}\sum_{a}\sum_{i,j} p_{a}(i,j) \tr[\sigma_{a}\state{\psi^{(i)}_{j}}] \nonumber \\ 
  &=& \frac{1}{Nd} \sum_{a}\sum_{i,j}\tr[M_{a}\state{\psi_{j}^{(i)}}] \tr[\sigma_{a}\state{\psi^{(i)}_{j}}]. \nonumber  
\end{eqnarray} 
The error rate $\mathbf{E}_{\cS} (\cM, \cR) $ is simply the average probability that Bob's measurement gives the wrong result: 
\begin{eqnarray} 
&& \mathbf{E}_{\cS} (\cM, \cR) \nonumber \\ 
&=& 1 - \frac{1}{Nd} \sum_{a}\sum_{i,j}\tr[M_{a}\state{\psi_{j}^{(i)}}] \tr[\sigma_{a}\state{\psi^{(i)}_{j}}] \nonumber \\ 
&=& 1 - F_{\cS}(\cM, \cR) \geq 1 - F^{\rm max}_{\cS} \nonumber \\
&\equiv& \cQ(A^{(1)}, A^{(2)}, \ldots, A^{(N)}). \end{eqnarray} 
The incompatibility $\cQ(A^{(1)},\ldots,A^{(N)})$ is thus the smallest error rate possible. 
\end{proof}

\section{Incompatibility of qubit observables}\label{sec:Q_qubitProof}

In this section, we prove Theorem.~\ref{thm:Q_qubit} by explicitly evaluating the incompatibility measures $\cQ(A,B)$ and $\cQ_{F}(A,B)$ for a pair of qubit observables $A = \alpha_{1}\Id + \alpha_{2}\vec{a}.\vec{\sigma}$ and $B = \beta_{1}\Id + \beta_{2}\vec{b}.\vec{\sigma}$.

\subsection{Evaluating $\cQ(A,B)$}

We first note that the ensemble $\cS$ comprising the eigenstates of a pair of qubit observables  $A, B$, can be denoted in terms of the vectors $\vec{a}, \vec{b}$ as follows:
\[ \cS \equiv \left\{\frac{\Id \pm \vec{a}\cdot\vec{\sigma}}{2}, \frac{\Id \pm \vec{b}\cdot\vec{\sigma}}{2}\right\}. \]  
To evaluate $\cQ$, we use the form of the accessible fidelity function in Eq.~\eqref{eq:accFid}. In order to evaluate the maximum fidelity of an ensemble of states in a $d$-dimensional Hilbert space, it suffices to optimize over POVMs with $d^{2}$ rank-one elements~\cite{davies:access, Fuchs-Sasaki}. Thus, for the case of qubit observables, it suffices to restrict our optimization to POVMs with four rank-one elements, that is,  $\cM = \{m_{i},|\chi_{i}\rangle\langle\chi_{i}|\}_{i=1}^{4}$, subject to $\sum_{i}m_{i}|\chi_{i}\rangle\langle\chi_{i}| = \Id$. We may parameterize the elements of $\cM$ in terms of vectors $c_{i}\in \mathbb{R}^{3}$, so that the optimization is over
\[ \cM \equiv \left\{m_{i}, \frac{\Id + \vec{c}_{i}\cdot\vec{\sigma}}{2}\right\}_{i=1}^4, \; |\vec{c}_{i}| = 1, \; \forall i=1,\ldots, 4,\]
subject to $\sum_{i=1}^{4} m_{i} = 1$ and $\sum_{i=1}^{4} m_{i} \vec{c}_{i} = \vec{0}$.

The maximum fidelity function $F^{\rm max}_{\cS}$ involves the following optimization problem:
\[ F^{\rm max}_\cS = \max_{\{m_{i}, \vec{c}_{i}\}} \sum_{i} m_{i} \lambda_\mx [C_{i}], \]
subject to $\sum_{i=1}^4 m_i = 1$ and $\sum_{i=1}^4 m_i \vec{c}_i = \vec{0}$, where the matrices $C_{i}$ are given by
\[ C_{i} = \frac{1}{2}\left[\Id + \dfrac{(\vec{c}_{i}\cdot \vec{a}) \vec{a} + (\vec{c}_{i}\cdot \vec{b})\vec{b}}{2}\cdot \vec{\sigma}\right]. \]

Let $\theta_{i}$ denote the angle $\vec{c}_{i}$ makes with the plane containing $\vec{a}$ and $\vec{b}$ and let its component on this plane make an angle of $\alpha_{i}$ with $\vec{a}$. That is,
\[ \vec{c}_{i} = \cos \theta_{i}\left(\cos \alpha_{i} \vec{a} + \sin \alpha_{i} \vec{a}_\perp \right) + \sin \theta_{i} \hat{e}, \]
where $\vec{a}_\perp$ is the vector perpendicular to $\vec{a}$ in the plane of $\vec{a}$ and $\vec{b}$ and $\hat{e}$ is the unit vector perpendicular to the plane. So the constraint on $\sum_{i}m_{i}c_{i}$ becomes
\[ \sum_{i=1}^{4} m_{i} \cos \theta_{i} \cos \alpha_{i} = \sum_{i=1}^{4} m_{i} \cos \theta_{i} \sin \alpha_{i} = \sum_{i=1}^4 m_{i} \sin \theta_{i} = 0. \]
The maximum eigenvalues of the matrices $C_{i}$ are given by,
\begin{eqnarray*}
\lambda_{\mx}[C_{i}] &=& \frac{1 + |(\vec{c}_i\cdot \vec{a})\vec{a} + (\vec{c}_i \cdot \vec{b})\vec{b}|}{2} \\
&=& \frac{1}{2} + \frac{1}{2}\frac{|\cos \theta_i|}{2} \sqrt{g(\alpha_{i},\delta)},
\end{eqnarray*}
\[g(\alpha_{i}, \delta) = \cos^2 \alpha_i + \cos^2(\alpha_i - \delta) + 2 \cos\alpha_i \cos(\alpha_i - \delta)\cos\delta .\]
The optimization problem now simplifies to,
\[
F^{\rm acc}_\cS = \frac{1}{2} + \frac{1}{4} \max_{\{m_i, \theta_i, \alpha_i\}}\sum_i m_i |\cos \theta_i|\sqrt{g(\alpha_{i},\delta)}.\]
The optimization over the parameters $\theta_i$ is trivial -- we should simply take each $\theta_i = 0$, which satisfies the constraint on the $\sin \theta_i$. Thus, all the measurement operators of the optimal POVM lie in the plane spanned by $\vec{a}$ and $\vec{b}$. 

The optimization further simplifies to
\[ f _\cS = \max_{\{m_i, \alpha_i\}}\sum_i m_i \sqrt{g(\alpha_{i},\delta)}, \]
subject to the constraints $\sum_i m_i = 1$ and $\sum_i m_i \cos \alpha_i = \sum_i m_i \sin\alpha_i = 0$. Note that even though our objective function is a convex sum, we cannot say that the maximal value will simply be equal to the maximum term, because the constraints on the sines and cosines cannot then both be satisfied. In fact, we have,
\[ f_\cS \le \max_\alpha \sqrt{g(\alpha_{i},\delta)} = \sqrt{\max_\alpha g(\alpha,\delta)}, \]
with equality holding in case there are multiple solutions $\alpha$ which maximize this function and we can take convex combination of them such that the additional constraints hold. To obtain the extremum of the expression under the square root, we require,
\[ \sin 2\alpha + \sin 2(\alpha - \delta) + 2 \cos \delta \sin(2\alpha - \delta) = 0. \]
The above equality can hold for a number of $\alpha$, including $\alpha = \delta/2, \pi/2 + \delta/2, \pi + \delta/2, 3\pi/2 + \delta/2$. When $\cos \delta > 0$, maxima are obtained at $\alpha = \delta/2, \pi + \delta/2$ and minima at $\alpha = \pi/2 + \delta/2, 3\pi/2 + \delta/2$; when $\cos \delta \le 0$, we have the reverse.

The maximum value in either case is given by $(1 + |\cos \delta|)^2$. Now if we take $m_1 = m_2 = 1/2$ and $\alpha_1 = \delta/2$, $\alpha_2 = \pi + \delta/2$, we can actually satisfy,
\[ m_1 \cos \alpha_1 + m_2 \cos \alpha_2 = m_1 \sin \alpha_1 + m_2 \sin \alpha_ 2 = 0.\]
So we can achieve
\begin{eqnarray}
f_\cS &=& 1 + |\cos \delta|. \nonumber \\
\Rightarrow F^{\acc}_\cS &=& \frac{3}{4} + \frac{1}{4}|\cos \delta|, \nonumber
\end{eqnarray}
with the maximum value attained for a two-outcome von Neumann measurement:
\begin{widetext}
\[ \cM = \begin{cases} \left\{\dfrac{1}{2}\left(\Id + \dfrac{\vec{a} + \vec{b}}{2 \cos(\delta/2)}\cdot \vec{\sigma}\right), \dfrac{1}{2}\left(\Id - \dfrac{\vec{a} + \vec{b}}{2 \cos(\delta/2)}\cdot \vec{\sigma}\right) \right\}, & \text{if }\ \vec{a}\cdot \vec{b} = \cos \delta \ge 0 \\
\left\{\dfrac{1}{2}\left(\Id + \dfrac{\vec{a} - \vec{b}}{2 \sin(\delta/2)}\cdot \vec{\sigma}\right), \dfrac{1}{2}\left(\Id - \dfrac{\vec{a} - \vec{b}}{2 \sin(\delta/2)}\cdot \vec{\sigma}\right) \right\}, & \text{if }\ \vec{a}\cdot \vec{b} = \cos \delta < 0.
\end{cases} \]
\end{widetext}

\subsection{Evaluating $\cQ_{F}(A,B)$}
Recall that $\cQ_{F}(A,B)$ is defined as,
\begin{align*}
Q_{F}(A\rightarrow B) & =\underset{\rho}{\text{sup}}\left\{ 1-F^{2}(Pr_{\rho}^{B}(j),Pr_{\rho}^{A\rightarrow B}(j))\right\} \\
 & =\underset{\rho}{\text{sup}}\left\{ 1-\left[\sum_{j}\sqrt{Pr_{\rho}^{B}(j)}\sqrt{Pr_{\rho}^{A\rightarrow B}(j)}\right]^{2}\right\}. 
\end{align*} 

Parametrizing $\rho$ in terms of a real vector $\vec{r} \in \mathbb{R}^{3}$, this simplifies to, 
\[ \cQ_{F}(A\rightarrow B) = \frac{1}{2} - \frac{1}{2}\underset{\vec{r}}{\text{min}}f_{\vec{a}, \vec{b}}(\vec{r}), \]
where, the function $f_{\vec{a}, \vec{b}}(\vec{r})$ is given by
\begin{eqnarray}
f_{\vec{a}, \vec{b}}(\vec{r}) &=& (\vec{r}\cdot\vec{a})(\vec{r}\cdot\vec{b})(\vec{a}\cdot\vec{b}) \nonumber \\
&-& \frac{1}{2}\sqrt{\left(1-\left((\vec{r}\cdot\vec{a})(\vec{a}\cdot\vec{b})\right)^{2}\right)\left(1-\left(\vec{r}\cdot\vec{b}\right)^{2}\right)} . \nonumber
\end{eqnarray}
Let $\theta$ denote the angle made by $\vec r$ with the plane defined by the vectors $\vec{a}, \vec{b}$, and $\alpha$ be
the angle made by the component of $\vec r$ with
$\vec a$ in the plane. Then we can rewrite the expression for $Q_{F}(A\rightarrow B)$ in terms of these angles as,
\[
Q_{F}(A\rightarrow B) =\frac{1}{2}-\frac{1}{2}\underset{\theta,\alpha}{\text{min}}f_{\delta}(\theta, \alpha),
\]
%

\begin{widetext}
\[
 f_{\delta} (\theta, \alpha) = \cos^{2}\theta\cos\alpha\cos\delta\cos\left(\alpha-\delta\right)+\sqrt{\left(1-\cos^{2}\theta\cos^{2}\alpha\cos^{2}\delta\right)\left(1-\cos^{2}\theta\cos^{2}\left(\alpha-\delta\right)\right)},\]
\end{widetext}
where, $\cos\delta= \vec a.\vec{b}$ as before. Taking the partial derivative with respect to $\theta$, we see that $\frac{\partial f}{\partial\theta}= 0$, iff $\theta=0$ or $\frac{\pi}{2}$. When $\theta=\frac{\pi}{2}$, the function $f_{\delta}(\theta, \alpha)$ attains its maximum value of $1$ for any value of $\delta$, thus indicating that $f_{\delta}$ attains its minimum value for $\theta=0$. 
\begin{figure}
\includegraphics[scale=0.25]{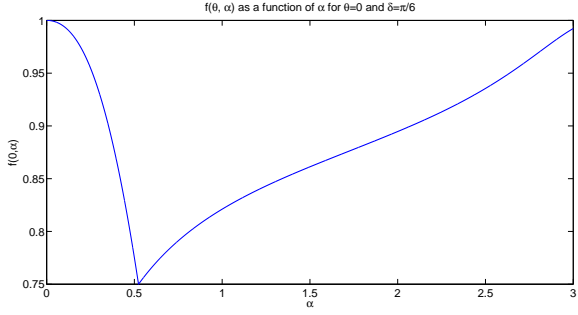}
\caption{$f_{\delta}(0, \alpha)$ for $\delta = \frac{\pi}{6}$.}
\end{figure}
Setting $\theta=0$ and plotting $f_{\delta}(0, \alpha)$ for a given value of $\delta$, we see that its minimum value is attained for $\alpha=\delta$. However, since the derivative $\frac{df_{\delta}(0,\alpha)}{d\alpha}$ is discontinuous at $\alpha = \delta$,  we can formally prove that $\alpha = \delta$ is indeed the global minimum as follows: we first consider an $\epsilon$-neighbourhood around $\alpha=\delta$ and see that it is indeed a local minimum for $f_{\delta}(0, \alpha)$. Then we can show that $f_{\delta}(0, \alpha)$ is monotonically decreasing for $\alpha < \delta$ and monotonically increasing for $\alpha > \delta$, thus proving that $\alpha=\delta$ is indeed the global minimum for $f_{\delta}(0, \alpha)$. 

In other words, the state that minimizes the fidelity function corresponds to $\vec{r}=\vec{b}$ (an eigenstate of $B$), so that
\[ Q_{F}(A\rightarrow B)=\frac{1}{2}-\frac{\left(\vec{a}\cdot\vec{b}\right)^{2}}{2}.\]  
A similar argument shows that 
\[ Q_{F}(B\rightarrow A)=\frac{1}{2}-\frac{\left(\vec{a}\cdot\vec{b}\right)^{2}}{2},\]
where the state that minimizes the fidelity function is an eigenstate of $A$. \\

\section{Additivity of accessible fidelity}\label{sec:additivity_proof}
Here we prove Lemma~\ref{lem:FAdditivity} on the additivity of the maximum fidelity function for ensembles $\cS_{1}\in\cH_{d_{1}}$ and $\cS_{2} \in \cH_{d_{2}}$. 
\begin{proof}
Consider the optimal POVM $\mathbb{M}_{1}$ that achieves the maximum fidelity $F^{\rm max}_{\cS_{1}}$ for ensemble $\cS_{1}$ and the optimal POVM $\mathbb{M}_{2}$ that achieves the  fidelity $F^{\rm max}_{\cS_{2}}$ for ensemble $\cS_{2}$. By combining the elements of $\mathbb{M}_{1}$ and $\mathbb{M}_{2}$ weighted appropriately, we have a POVM $\mathbb{M}$ that acts on the entire space $\cH_{d} \equiv \cH_{d_{1}}\oplus\cH_{d_{2}}$. Clearly, the fidelity achieved by $\mathbb{M}$ constitutes a lower bound for the maximum fidelity $F^{\rm max}_{\cS}$ for the ensemble $\cS=\cS_{1}\oplus\cS_{2}$. Therefore, we have,
\[ F^{\rm max}_{\cS} \geq \frac{1}{N_{1}d_{1} + N_{2}d_{2}}\left(N_{1}d_{1} F^{\rm max}_{\cS_{1}} + N_{2}d_{2}F^{\rm max}_{\cS_{2}}\right). \]

We now prove additivity by showing that the maximum fidelity function $F^{\rm max}_{\cS}$ is also upper bounded by the weighted average of $F^{\rm}_{\cS_{1}}$ and $F^{\rm max}_{\cS_{2}}$. Note that for any vector $\ket{\phi_a} \in \cH_{d_1}\oplus\cH_{d_2}$ in the optimal rank-one POVM $\{\chi_a, \state{\phi_a}\}$ that attains $F^{\rm max}_{\cS_{1}\oplus\cS_{2}}$, can be written as,
\begin{eqnarray}
 \state{\phi_a} &=& P_1\state{\phi_a}P_1 + P_1\state{\phi_a}P_2 \nonumber \\
&& + P_2\state{\phi_a}P_1 + P_2\state{\phi_a}P_2, \nonumber
\end{eqnarray}
where $P_1$ and $P_2$ are the projectors on to the subspaces $\cH_{d_1}$ and $\cH_{d_2}$ respectively. Now, the condition $\sum_a \chi_a \state{\phi_a} = I_{d_1 + d_2}/(d_1 + d_2)$ implies,
\begin{align*}
  \sum_a \chi_a P_1\state{\phi_a}P_1 &= \frac{I_{d_1}}{d_1 + d_2}, \\
   \sum_a \chi_a P_2\state{\phi_a}P_2 &= \frac{I_{d_2}}{d_1 + d_2}.
\end{align*}
Therefore, we may bound $F^{\rm max}_{\cS}$ as,
\begin{widetext}
\begin{eqnarray*}
F^{\rm max}_{\cS} & = & \frac{d_1 + d_2}{N_1d_1 + N_2d_2} \max_{\substack{\left\{\chi_a, \ket{\phi_a}: \right. \\ \left. \sum_a \chi_a \state{\phi_a} =  \Id/(d_1 + d_2) \right\}}} \sum_a \chi_a \lambda_\mx\left[\sum_{\ket{\psi}_1 \in \cH(d_1)}|\braket{\phi_a}{\psi}_1|^2 \state{\psi}_1 + \sum_{\ket{\psi}_2 \in \cH(d_2)}|\braket{\phi_a}{\psi}_2|^2 \state{\psi}_2\right] \\
                             & \leq & \frac{d_1 + d_2}{N_1d_1 + N_2d_2} \max_{\{\chi_a, P_1\ket{\phi_a}\}} \sum_a \chi_a \lambda_\mx\left[\sum_{\ket{\psi}_1 \in \cH(d_1)}|\matel{\phi_a}{P_1}{\psi}_1|^2 \state{\psi}_1\right] \\
&& \quad + \quad \frac{d_1 + d_2}{N_1d_1 + N_2d_2} \max_{\{\chi_a, P_2\ket{\phi_a}\}} \sum_a \chi_a \lambda_\mx\left[\sum_{\ket{\psi}_2 \in \cH(d_2)}|\matel{\phi_a}{P_2}{\psi}_2|^2 \state{\psi}_2\right].
\end{eqnarray*}
\end{widetext}
We define the states $\ket{\phi_{1, a}} = \frac{P_1\ket{\phi_a}}{|\matel{\phi_1}{P_1}{\phi_a}|^2}$, $\ket{\phi_{2, a}} = \frac{P_2\ket{\phi_a}}{|\matel{\phi_a}{P_2}{\phi_a}|^2}$, which satisfy $||\ket{\phi_{1, a}}|| = ||\ket{\phi_{2, a}}|| = 1$, and the scalar coefficients
\[\chi_{i, a} = \frac{d_1 + d_2}{d_i}\matel{\phi_a}{P_i}{\phi_a}|^2\chi_a \]
with $i = 1, 2$, which are weights in $\cH_{d_1}$ and $\cH_{d_2}$ respectively. We can then restate the optimization in terms of the POVMs $\{\chi_{1,a}, \ket{\phi_{1,a}}\}$, $\{\chi_{2,a}, \ket{\phi_{2,a}}\}$ subject to the constraints $\sum_a \chi_{1,a} \state{\phi_{1,a}} =  I_{d_1}/d_1$ and $\sum_a \chi_{2,a} \state{\phi_{2,a}} =  I_{d_2}/d_2$, as follows.
\begin{widetext}
\begin{eqnarray*}
F^{\rm max}_{\cS} & \le & \frac{d_1}{N_1d_1 + N_2d_2}  \max_{\{\chi_{1,a}, \ket{\phi_{1,a}}\}} \sum_a\chi_{1,a}\lambda_\mx\left[\sum_{\ket{\psi}_1 \in \cH(d_1)}|\braket{\phi_{1,a}}{\psi}_1|^2 \state{\psi}_1\right]  \\
                                & + &  \frac{d_2}{N_1d_1 + N_2d_2} \max_{\{\chi_{2,a}, \ket{\phi_{2,a}}\}}\sum_a\chi_{2, a} \lambda_\mx\left[\sum_{\ket{\psi}_2 \in \cH(d_2)}|\braket{\phi_{2,a}}{\psi}_2|^2 \state{\psi}_2\right]  \\
                                & = & \frac{1}{N_1d_1 + N_2d_2}\left(N_1d_1 F^{\rm max}_{\cS_1} + N_2d_2 F^{\rm max}_{\cS_2}\right) .
\end{eqnarray*}
\end{widetext}
\end{proof}

\end{document}